\documentclass[runningheads,a4paper]{lmcs}
\pdfoutput=1

\usepackage{amssymb}
\usepackage{enumerate}
\usepackage[colorlinks=true]{hyperref}
\usepackage{tikz}
\usepackage{xcolor,latexsym,amsmath,extarrows,alltt}
\usepackage{xspace}
\usepackage{booktabs}
\usepackage{mathtools}
\usepackage{enumitem}
\usepackage{stmaryrd}
\usepackage{microtype}
\usepackage{bussproofs}
\usepackage{multirow}
\usepackage{wrapfig}

\theoremstyle{theorem}\newtheorem{theorem}{Theorem}
\theoremstyle{theorem}\newtheorem{lemma}[theorem]{Lemma}
\theoremstyle{theorem}
\theoremstyle{theorem}\newtheorem{proposition}[theorem]{Proposition}
\theoremstyle{definition}\newtheorem{definition}[theorem]{Definition}
\theoremstyle{definition}\newtheorem{algorithm}[theorem]{Algorithm}
\theoremstyle{definition}\newtheorem{example}[theorem]{Example}
\theoremstyle{definition}

\newcommand{\N}{\mathbb{N}}

\newcommand{\B}{\mathcal{B}}
\newcommand{\F}{\mathcal{F}}
\newcommand{\M}{\mathcal{M}}
\newcommand{\V}{\mathcal{V}}

\newcommand{\Constructors}{\mathcal{C}}
\newcommand{\Defineds}{\mathcal{D}}
\newcommand{\Var}{\mathit{Var}}
\newcommand{\Data}{\texttt{Data}}

\newcommand{\Card}{\mathsf{Card}}

\newcommand{\numrep}[1]{[#1]}
\newcommand{\ntimecomp}[1]{\ensuremath{\textrm{NTIME}\left(#1\right)}}
\newcommand{\exptime}[1]{\mathsf{EXP}^{#1}\mathsf{TIME}}

\newcommand{\asort}{\iota}
\newcommand{\asortorpair}{\kappa}
\newcommand{\atype}{\sigma}
\newcommand{\btype}{\tau}

\newcommand{\aindex}{l}
\newcommand{\bindex}{p}
\newcommand{\cindex}{q}
\newcommand{\dindex}{r}
\newcommand{\prog}{\mathsf{p}}

\newcommand{\arity}{\mathtt{arity}_\prog}
\newcommand{\vdashcall}{\vdash^{\mathtt{call}}}

\newcommand{\vvdash}{\Vdash}
\newcommand{\vvdashcall}{\Vdash^{\mathtt{call}}}

\newcommand{\apps}[3]{#1\ #2 \cdots #3}

\newcommand{\symb}[1]{\mathtt{#1}}
\newcommand{\interpret}[1]{\llbracket #1 \rrbracket_\B}

\newcommand{\down}[2]{#1\!\Downarrow\!#2}
\newcommand{\arrtype}{\Rightarrow}
\newcommand{\arrr}{\to}
\newcommand{\too}{\Rightarrow}
\newcommand{\blank}{\textbf{\textvisiblespace}}
\newcommand{\nul}{\symb{0}}

\newcommand{\nil}{\symb{[]}}
\newcommand{\cons}{\symb{::}}
\newcommand{\strue}{\symb{true}}
\newcommand{\sfalse}{\symb{false}}
\newcommand{\suc}{\symb{s}}
\newcommand{\bits}{\symb{list}}
\newcommand{\bool}{\symb{bool}}

\newcommand{\msymbol}{\symb{symbol}}
\newcommand{\mstate}{\symb{state}}
\newcommand{\mdirec}{\symb{direc}}

\newcommand{\ifte}[3]{\mathtt{if} \, #1 \, \mathtt{then} \, #2 \,
                    \mathtt{else} \, #3}
\newcommand{\choice}{\mathtt{choose}}
\newcommand{\identifier}[1]{\mathtt{#1}}
\newcommand{\progeval}[2]{\llbracket #1 \rrbracket(#2)}

\newcommand{\progresult}{\llbracket \prog \rrbracket(d_1,\dots,d_M)
\mapsto b}
\newcommand{\pclass}{\mathsf{P}}
\newcommand{\npclass}{\mathsf{NP}}
\newcommand{\neclass}{\mathsf{NEXP}}
\newcommand{\netime}[1]{\mathsf{NEXP}^{(#1)}}
\newcommand{\etime}[1]{\mathsf{EXP}^{(#1)}}
\newcommand{\expclass}{\mathsf{EXP}}

\newcommand{\elementary}{\mathsf{ELEMENTARY}}
\newcommand{\complexityfun}{h}

\newcommand{\typeorder}[1]{\ensuremath{\mathit{ord}\!\left(#1\right)}}

\begin{document}

\title{Cons-free programming with immutable functions}
\author[C.~Kop]{Cynthia Kop}
\address{Department of Computer Science, Copenhagen University $\star$}
\email{kop@di.ku.dk}
\thanks{$\star$ Supported by the Marie Sk{\l}odowska-Curie
action ``HORIP'', program H2020-MSCA-IF-2014, 658162.}

\begin{abstract}
We investigate the power of non-determinism in purely functional
programming languages with higher-order types.
Specifically, we set out to characterise the hierarchy
$$
\npclass \subsetneq \neclass \subsetneq \netime{2} \subsetneq \cdots
\subsetneq \netime{k} \subsetneq \cdots
$$
solely in terms of higher-typed, purely functional programs.
Although the work is incomplete, we present an initial approach using
\emph{cons-free programs with immutable functions}.
%The three main technical artifacts are that programs be
%\emph{cons-free}, \emph{terminating} and \emph{order $k+1$ immutable}.
\end{abstract}

\maketitle

\section{Introduction}

In~\cite{jon:01}, Jones introduces \emph{cons-free programming}:
working with a small functional programming language, cons-free
programs are defined to be \emph{read-only}: recursive data cannot
be created or altered (beyond taking sub-expressions), only read from
the input.  By imposing further restrictions on data order and
recursion style, classes of cons-free programs turn out to characterise
various classes in the time and space hierarchies of computational
complexity.  However, this concerns only \emph{deterministic} classes.

It is tantalising to consider the non-deterministic classes such as
$\npclass$: is there a way to characterise these using cons-free
programs?  Unfortunately, merely adding non-determinism to Jones'
language does not suffice: cons-free programs with data order $0$
characterise $\pclass$ whether or not a non-deterministic choice
operator is included in the language~\cite{bon:06}, and for higher
data orders $K$ adding such an operator increases the expressivity
from $\exptime{K}$ to $\elementary$~\cite{kop:sim:17}.
Thus, additional language features or limitations are needed.

In this work, we explore cons-free programs with \emph{immutable
functions}, where data of higher type may be created but not
manipulated afterwards.  We present some initial ideas to obtain both
a characterisation of $\npclass$ by terminating cons-free programs
with immutable functions, and a generalisation towards all classes in
the hierarchy
$$
\npclass \subsetneq \neclass \subsetneq \netime{2} \subsetneq \cdots
\subsetneq \netime{k} \subsetneq \cdots
$$

This is a work in progress; the core results have not yet been fully
proven correct, and definitions may be tweaked.  The goal is not only
to find a characterisation of the hierarchy above, but also to
identify the difficulties on the way.  It is not unlikely that this
could help to find other interesting characterisations, and highlight
properties of non-determinism.

\section{Preliminaries}

\emph{A more elaborate presentation of the subjects discussed in this
section is available in~\cite{kop:sim:17}.}

\subsection{Turing Machines and complexity}

We assume familiarity with standard notions of Turing Machines and
complexity classes (see, e.g.,
\cite{Papadimitriou:complexity,Jones:CompComp});
here, we fix notation.

\medskip
Turing Machines (TMs) are triples $(A,S,T)$ of finite sets of
\emph{tape symbols}, \emph{states} and \emph{transitions}, where
$A \supseteq \{0,1,\blank\}$,\ 
$S \supseteq \{ \symb{start}, \symb{accept}, \symb{reject} \}$ and
$T$ contains tuples $(i,r,w,d,j)$ with $i \in S \setminus
\{\symb{accept},\symb{reject}\}$ (the \emph{original state}), $r \in
A$ (the \emph{read symbol}), $w
\in A$ (the \emph{written symbol}), $d \in \{ \symb{L},\symb{R} \}$
(the \emph{direction}), and $j \in S$ (the \emph{result state}).
Every TM in this paper has a single, right-infinite tape.
A TM \emph{accepts} a decision problem $X \subseteq \{0,1\}^+$ if
for any $x \in \{0,1\}^+$: $x \in X$ if{f} there is an evaluation
starting on the tape $\blank x_1\dots x_n\blank\blank\dots$ which
ends in the $\symb{accept}$ state.
For $\complexityfun : \N \longrightarrow \N$ a function, a TM
$\mathcal{M}$ \emph{runs in time} $\lambda n.\complexityfun(n)$ if
for $x \in \{0,1\}^n$: if $\mathcal{M}$ accepts $x$, then this can be
done in at most $\complexityfun(n)$ steps.

\medskip
Let $\complexityfun : \N \rightarrow \N$ be a function.
Then, $\ntimecomp{\complexityfun(n)}$ is the set of all $X \subseteq
\{0,1\}^+$ such that there exist $a > 0$ and a TM
running in time $\lambda n.a \cdot \complexityfun(n)$ that accepts
$X$.

For $K,n \geq 0$, let $\exp_2^0(n) = n$ and $\exp_2^{K+1}(n) =
\exp_2^K(2^n) = 2^{\exp_2^K(n)}$.
For $K \geq 0$, define
$
\netime{K} \triangleq \bigcup_{a,b \in \N}
\ntimecomp{\textrm{exp}_2^{K}(an^b)}
$.
Let $\elementary \triangleq \bigcup_{K \in \N} \netime{K}$.

\subsection{Non-deterministic programs}\label{subsec:programs}

We consider functional programs with simple types and call-by-value
evaluation.  Data constructors (denoted $\identifier{c}$) are at
least $\strue,\sfalse : \bool$,
$\nil : \bits$ and $\cons : \bool \arrtype \bits \arrtype \bits$
(denoted infix), although others are allowed.
\begin{wrapfigure}{r}{7.7cm}
\vspace{-12pt}
\fbox{
\begin{tabular}{rcl}
$d,b \in \texttt{Data}$ & ::= & $\apps{\identifier{c}}{d_1}{d_m} \mid
  (d,b)$ \\
$v,w \in \texttt{Value}$ & ::= & $d \mid (v,w) \mid
\apps{\identifier{f}}{v_1}{v_n}$ \\
& & ($n < \arity(\identifier{f})$) \\
\end{tabular}
}
\vspace{-12pt}
\end{wrapfigure}
There is no pre-defined integer datatype.
Notions of \emph{data} and \emph{values} are given by the grammar to
the right, where $\identifier{f}$ indicates a function symbol defined
by one or more clauses.
The language supports $\ifte{}{}{}$ statements, but has no
$\symb{let}$ construct.
For a program $\prog$ with \emph{main function} $\identifier{f}_1 :
\asort_1 \arrtype \dots \arrtype \asort_M \arrtype \asortorpair$---%
where $\asortorpair$ and each $\asort_i$ have type order $0$---and
\emph{input data} $d_1,\dots,d_M$, $\prog$ has \emph{result value}
$b$ if there is an evaluation $\apps{\identifier{f}_1}{d_1}{d_M} \arrr
b$.

There is also a non-deterministic choice construct:
$\apps{\choice}{s_1}{s_m}$ may evaluate to a value $v$ if some $s_i$
does.  Thus, a program can have multiple result values on the same
input.

\medskip
A program $\prog$ with main function $\identifier{f}_1 : \bits \arrtype
\bool$ \emph{accepts} a decision problem $X$ if for all $x = x_1\dots
x_n \in \{0,1\}^*$: $x \in X$ if{f} $\prog$ has result value $\strue$
on input $\overline{x_1}\cons\dots\cons\overline{x_n}\cons\nil$, where
$\overline{1} = \strue$ and $\overline{0} = \sfalse$.  It is not
necessary for $\strue$ to be the \emph{only} result value.

\subsection{Cons-free programs}

A clause $\apps{\identifier{f}}{\ell_1}{\ell_k} = s$ is
\emph{cons-free} if for all sub-expressions $t$ of $s$: if $t =
\apps{\identifier{c}}{t_1}{t_m}$ with $\identifier{c}$ a data
constructor, then $t \in \Data$ or $t$ also occurs as a sub-expression
of some $\ell_i$.
A program is cons-free if all its clauses are.

Intuitively, in a cons-free program no new recursive data can be
created: all data encountered during evaluation occur inside the
input, or as part of some clause.

\begin{example}
The clauses for $\symb{last}$ below are cons-free; the clauses for
$\symb{flip}$ are not. \\
\[
\begin{array}{ll}
\symb{last}\ (x\cons \nil) = x &
\symb{flip}\ \nil = \nil \\
\symb{last}\ (x\cons y\cons zs) = \symb{last}\ (y\cons zs) &
\symb{flip}\ (\strue\cons xs) = \sfalse\cons (\symb{flip}\ xs) \\
& \symb{flip}\ (\sfalse\cons xs) = \strue\cons (\symb{flip}\ xs) \\
\end{array}
\]
\end{example}

\subsection{Counting}\label{subsec:counting}
Cons-free programs neither have an integer data
type, nor a way to construct unbounded recursive data (e.g., we cannot
build $\nul,\ \suc\ \nul,\ \suc\ (\suc\ \nul)$ etc.).  It \emph{is},
however, possible to design cons-free programs operating on certain
bounded classes of numbers: by representing numbers as values using
the input data.  For example, given a list $cs$ of length $n$ as
input:
\begin{enumerate}
\item\label{count:list}
  numbers $i \in \{0,\dots,n\}$ can be represented as lists of
  length $i$ (sub-expressions of $cs$);
\item\label{count:base}
  numbers $i \in \{0,\dots,4 \cdot (n+1)^2-1\}$ can be represented
  as tuples $(l_1,l_2,l_3) : \bits \times \bits \times \bits$%, where
  %$l_2$ and $l_3$ are sub-expressions of $cs$ and $l_1$ is a
  %sub-expression of $\strue\cons\strue\cons\strue\cons\nil$: here,
  :
  writing $i = k_1 \cdot (n+1)^2 + k_2 \cdot (n+1) + k_3$, the number
  $i$ is represented by a tuple $(l_1,\dots,l_3)$ such that each $l_i$
  has length $k_i$;
\item\label{count:exp}
  numbers $i \in \{0,\dots,2^{4 \cdots (n+1)^2}-1\}$ can be
  represented by values $v : (\bits \times \bits \times \bits)
  \arrtype \bool$: writing $i_0\dots i_{4 \cdot (n+1)^2-1}$ for the
  bitvector corresponding to $i$, it is represented by any value $v$
  such that $v\ [j] \arrr \strue$ if{f} $i_j = 1$, where $[j]$ is the
  representation of $j \in \{0,\dots,4 \cdot (n+1)^2-1\}$ as a tuple
  from point (\ref{count:base}).
\end{enumerate}

Building on (\ref{count:exp}), numbers in $\{0,\dots,\exp_2^K(a n^b)
\}$ can be represented by values of type order $K$.  It is not hard to
construct cons-free rules to calculate successor and predecessor
functions, and to test whether a number representation corresponds to
$0$.

Jones~\cite{jon:01} uses these number representations and counting
functions to write a cons-free program with data order $K$ which
simulates a given TM running in at most $\exp_2^K(a n^b)$ steps.
However, this program relies heavily on the machine being
deterministic.

\section{Characterising $\npclass$}

To characterise non-deterministic classes we start by countering this
problem: we present a cons-free program which determines the final
state of a non-deterministic TM running
in $\lambda n.\complexityfun(n)$ steps, given number representations
for $i \in \{0,\dots,\complexityfun(n)\}$ and counting functions.

For a machine (A,S,T), let $C$ be a fixed number such that for every
$i \in S$ and $r \in A$ there are at most $C$ different triples
$(w,d,j)$ such that $(i,r,w,d,j) \in T$.
Let $T'$ be a set of tuples $(k,i,r,w,d,j)$ such that (a)
$T = \{ (i,r,w,d,j) \mid (k,i,r,w,d,j)$ occurs in $T' \}$ and (b)
each combination $(k,i,r)$ occurs at most once in $T'$.
Now, our simulation uses the data constructors:
$\strue : \bool,\ \sfalse : \bool,\ \nil : \bits$ and $\cons :
  \bool \arrtype \bits \arrtype \bits$ noted in 
  Section~\ref{subsec:programs};
$\symb{a} : \msymbol$ for $a \in A$ (writing $\symb{B}$ for the
  blank symbol), $\symb{L},\symb{R} : \mdirec$ and
  $\symb{s} : \mstate$ for $s \in S$;
$\symb{action} : \msymbol \arrtype \mdirec \arrtype \mstate
  \arrtype \symb{trans}$;
$\symb{x1} : \symb{option}$, \dots, $\symb{xC} : \symb{option}$;
and
$\symb{end} : \symb{state} \arrtype \symb{trans}$.
The rules to simulate the machine are given in
Figure~\ref{fig:machine}.

\newcommand{\extend}[1]{\multicolumn{2}{l}{#1}}

\begin{figure}[!ht]
\begin{tabular}{ll}
\extend{$\symb{run}\ cs = \symb{test}\ (\symb{state}\ cs\ 
(\symb{rndf}\ cs\ \numrep{\complexityfun(|cs|)})\ 
\numrep{\complexityfun(|cs|)})$} \\
$\symb{test}\ \symb{accept} = \strue$ \\
$\symb{test}\ s = \sfalse$ &
  for all $s \in S \setminus \{ \symb{accept} \}$ \\
\vspace{-6pt} \\
\extend{$\symb{rndf}\ cs\ \numrep{n} = \symb{rnf}\ cs\ 
  (\apps{\choice}{\symb{x1}}{\symb{xC}})\ \numrep{n}$} \\
\extend{$\symb{rnf}\ cs\ x\ \numrep{n} =
  \ifte{\:\numrep{n=0}\:}{\:x\:}{\:\symb{cmp}\ x\ (\symb{rndf}\ 
  cs\ \numrep{n-1})\ \numrep{n}}$} \\
\extend{$\symb{cmp}\ x\ H\ \numrep{n}\ \numrep{i} =
  \ifte{\:\numrep{n=i}\:}{\:x\:}{\:F\ i}$} \\
\vspace{-6pt} \\
$\symb{transition}\ \symb{i}\ \symb{r}\ \symb{ch}_k =
  \symb{action}\ \symb{w}\ \symb{d}\ \symb{j}$ &
  for all $(k,i,r,w,d,j) \in T'$ \\
$\symb{transition}\ \symb{i}\ x\ y = \symb{end}\ \symb{i}$ &
  for $i \in \{ \symb{accept},\symb{reject} \}$ \\
$\symb{transition}\ \symb{i}\ \symb{r}\ \symb{ch}_k =
  \symb{end}\ \symb{reject}$ &
  for all $(k,i,r)$ s.t.~$i \notin \{ \symb{accept},\symb{reject} \}$
  and \\
  & there are no $(w,d,j)$ with $(k,i,r,w,d,j) \in T'$ \\
\extend{$\symb{transat}\ cs\ H\ \numrep{n}\ =
  \symb{transition}\ (\symb{state}\ cs\ H\ \numrep{n})\ 
  (\symb{tapesymb}\ cs\ H\ \numrep{n})\ 
  (H\ \numrep{n})$} \\
\vspace{-6pt} \\
$\symb{get1}\ (\symb{action}\ x\ y\ z) = x$ &
$\symb{get1}\ (\symb{end}\ x) = \symb{B}$ \\
$\symb{get2}\ (\symb{action}\ x\ y\ z) = y$ &
$\symb{get2}\ (\symb{end}\ x) = \symb{R}$ \\
$\symb{get3}\ (\symb{action}\ x\ y\ z) = z$ &
$\symb{get3}\ (\symb{end}\ x) = x$ \\
\vspace{-6pt} \\
\extend{$\symb{state}\ cs\ H\ \numrep{n} = \ifte{\:\numrep{n = 0}\:
  }{\:\symb{start}\:}{\:\symb{get3}\ (\symb{transat}\ cs\ H\ 
  \numrep{n-1})}$} \\
\vspace{-6pt} \\
\extend{$\symb{tapesymb}\ cs\ H\ \numrep{n} = \symb{tape}\ cs\ H\ 
\numrep{n}\ (\symb{pos}\ cs\ H\ \numrep{n})$} \\
\vspace{-6pt} \\
$\symb{pos}\ cs\ H\ \numrep{n} = \symb{if}\ \numrep{n=0}\ 
  \symb{then}\ \numrep{0}$ \\
\extend{$\phantom{\symb{pos}\ cs\ H\ \numrep{n} =}~
  \symb{else}\ \symb{adjust}\ cs\ (\symb{pos}\ 
  cs\ H\ \numrep{n-1})\ 
  (\symb{get2}\ (\symb{transat}\ cs\ H\ \numrep{n-1}))$} \\
$\symb{adjust}\ cs\ \numrep{p}\ \symb{L} = \numrep{p-1}$ \\
$\symb{adjust}\ cs\ \numrep{p}\ \symb{R} = \numrep{p+1}$ \\
\vspace{-6pt} \\
\extend{$\symb{tape}\ cs\ H\ \numrep{n}\ \numrep{p} =
  \symb{if}\ \numrep{n = 0}\ \symb{then}\ \symb{inputtape}\ cs\ 
  \numrep{p}$} \\
\extend{$\phantom{\symb{tape}\ cs\ H\ \numrep{n}\ \numrep{p} =}~
  \symb{else}\ \symb{tapehelp}\ cs\ H\ \numrep{n}\ \numrep{p}\ 
  (\symb{pos}\ cs\ H\ \numrep{n-1})$} \\
\extend{$\symb{tapehelp}\ cs\ H\ \numrep{n}\ \numrep{p}\ \numrep{i}
  = \symb{if}\ \numrep{p=i}\ \symb{then}\ \symb{get1}\ 
  (\symb{transat}\ cs\ H\ \numrep{n-1})$} \\
\extend{$\phantom{\symb{tapehelp}\ cs\ H\ \numrep{n}\ \numrep{p}\ 
  \numrep{i} =}~\symb{else}\ \symb{tape}\ cs\ H\ \numrep{n-1}\ 
  \numrep{p}\ \numrep{i}$} \\
\vspace{-6pt} \\
\extend{$\symb{inputtape}\ cs\ \numrep{p} =
  \ifte{\:\numrep{p=0}\:}{\:\symb{B}\:}{\:
  \symb{nth}\ cs\ \numrep{p-1}}$} \\
$\symb{nth}\ \nil\ \numrep{p} = \symb{B}$ \\
\extend{$\symb{nth}\ (x\cons xs)\ \numrep{p} = \ifte{\numrep{p = 0}}{
  \symb{bit}\ x}{\symb{nth}\ xs\ \numrep{p-1}}$} \\
$\symb{bit}\ \strue = \symb{1}$ \\
$\symb{bit}\ \sfalse = \symb{0}$ \\
\end{tabular}
\vspace{-6pt}
\caption{Simulating a Non-deterministic Turing Machine $(A,S,T')$}
\label{fig:machine}
\end{figure}

\bigskip
If $\complexityfun$ is a polynomial, this program has data order $1$
following (\ref{count:base}) of Section~\ref{subsec:counting}.
Thus, non-deterministic cons-free programs with data order $1$ can
accept any decision problem in $\npclass$.
However, since even deterministic such programs can accept all
problems in $\expclass \supseteq \npclass$, this is not surprising.
What \emph{is} noteworthy is how we use the higher-order value: there
is just one functional variable, which, once it has been created, is
passed around but never altered.  This is unlike the values
representing numbers, where we modify a functional value by
taking its ``successor'' or ``predecessor''.
This observation leads to the following definition:

\begin{definition}
A program has \emph{immutable functions} if for all clauses
$\apps{\identifier{f}}{\ell_1}{\ell_k} = s$: (a) the clause uses at
most one variable with a type of order $>0$, and (b) if there is such
a variable, then $s$ contains no other sub-expressions with type order
$>0$.
\end{definition}

Every program where all variables have type order $0$ automatically
has immutable functions.  The program of Figure~\ref{fig:machine} also
has immutable functions, provided the number representations
$\numrep{n}$ all have type order $0$.
%Note that the clauses for
%$\symb{run}$ and $\symb{rndf}$ and $\symb{rnf}$ all have
%sub-expressions of a higher type in their right-hand sides, but this
%is acceptable because they do not use higher-order variables.
We thus conclude:

\begin{lemma}\label{lem:simulmachine}
Every decision problem in $\npclass$ is accepted by a terminating
cons-free program with immutable functions.
\end{lemma}

\begin{proof}
If $X \in \npclass$, then there are $a,b$ and a TM $\M$ such that for
all $n$ and $x \in \{0,1\}^n$:
$x \in X$ if{f} there is an evaluation of $\M$ which accepts $x$ in
at most $a \cdot n^b$ steps.
Using $\M$ to build the program of Figure~\ref{fig:machine},
$\symb{run}\ \overline{x} \arrr \strue$ if{f} $\M$ accepts $x$, if{f}
$x \in X$.
%We are done because this program is cons-free, terminating
%and has immutable functions.
\end{proof}

For terminating cons-free programs with immutable functions to
\emph{characterise} $\npclass$, it remains to be seen that every
decision problem accepted by such a program is in $\npclass$.  That
is, for a fixed program $\prog$ we must design a (non-deterministic)
algorithm operating in polynomial time which returns $\symb{YES}$ for
input data $d$ if{f} $\prog$ has $\strue$ as a result value on input
$d$.

Towards this purpose, we first alter the program slightly:

\begin{lemma}\label{lem:progtransform}
If $\progresult$, then $\progeval{\prog'}{d_1,\dots,d_m} \mapsto b$,
where $\prog'$ is obtained from $\prog$ by replacing all
sub-expressions $\apps{s}{t_1}{t_n}$ in the right-hand side of a
clause where $s$ is not an application by $s \circ t_1 \cdots t_n$.
Here,
\begin{itemize}
\item $(\ifte{s_1}{s_2}{s_3}) \circ t_1 \cdots t_n =
  \ifte{s_1}{(s_2 \circ t_1 \cdots t_n)}{(s_3 \circ t_1 \cdots t_n)}$
\item $(\apps{\choice}{s_1}{s_m}) \circ t_1 \cdots t_n =
  \apps{\choice}{(s_1 \circ t_1 \cdots t_n)}{(s_m \circ t_1 \cdots
  t_n)}$
\item $(\apps{a}{s_1}{s_i}) \circ t_1 \cdots t_n =
  \apps{\apps{a}{s_1}{s_i}}{t_1}{t_n}$ for $a \in \V \cup
  \Constructors \cup \Defineds$.
\end{itemize}
In addition, there is a transformation which preserves and reflects
$\progresult$ such that all argument types of defined symbols and
all clauses have type order $\leq 1$ and there are no clauses
$\apps{\identifier{f}}{\ell_1}{\ell_k} = x$ with $x$ a variable of
functional type.
\end{lemma}

\begin{proof}[Idea]
The $\symb{if}$/$\choice$ change trivially works, the type order
change uses the restriction on data order as detailed
in~\cite[Lemma 1]{kop:sim:17} and functional variables can be avoided
because, by immutability, all clauses for $\identifier{f}$ must have
a variable as their right-hand side.
\end{proof}

We implicitly assume that the transformations of
Lemma~\ref{lem:progtransform} have been done.

To reason on values in the algorithm, we define a semantical way to
describe them.

\begin{definition}
For a fixed cons-free program $\prog$ and input data expressions
$d_1,\dots,d_M$, let $\B$ be the set of data expressions which occur
either as a sub-expression of some $d_i$, or as sub-expression of the
right-hand side of some clause.  For $\asort$ a sort (basic type),
let $\interpret{\asort} := \{ b \in \B \mid b : \asort \}$.
Also let
$\interpret{\atype \times \btype} := \interpret{\atype} \times
  \interpret{\btype}$, and if $\asortorpair$ is not an arrow type,
$\interpret{\atype_1 \arrtype \dots \arrtype \atype_n \arrtype
  \asortorpair} := \{ (e_1,\dots,e_n,o) \mid \forall 1 \leq i \leq n
  [e_i \in \interpret{\atype_i}] \wedge o \in \interpret{\asortorpair}
  \}$
\end{definition}

By relating values of type $\atype$ to elements of
$\interpret{\atype}$, we can prove that Algorithm~\ref{alg:base} below
returns $b$ if and only if $b$ is a result value of $\prog$ for input
$d_1,\dots,d_M$.

\begin{algorithm}\label{alg:base}
Let $\prog$ be a fixed, terminating cons-free program with immutable
functions.

{\bf Input:} data expressions $d_1 : \asortorpair_1,\dots,d_M :
\asortorpair_M$.

For every function symbol $\identifier{f}$ with type $\atype_1
\arrtype \dots \arrtype \atype_m \arrtype \asortorpair$ and arity $k$
(the number of arguments to $\identifier{f}$ in clauses), every
$k \leq i \leq m$ and all $e_1 \in \interpret{\atype_1},\dots,e_i \in
\interpret{\atype_i},o \in \interpret{\atype_{i+1} \arrtype \dots
\arrtype \atype_m \arrtype \asortorpair}$, note down a ``statement''
$\vdash \apps{\identifier{f}}{e_1}{e_i} \mapsto o$.  For all clauses
$\apps{\identifier{f}}{\ell_1}{\ell_k} = s$, also note down
statements $\eta \vdash t \arrr o$ for every sub-expression $t :
\atype$ of $s \circ x_{k+1} \cdots x_i$, $o \in \interpret{\atype}$
and $\eta$ mapping all $y : \btype \in \Var(\apps{\apps{\identifier{f
}}{\ell_1}{\ell_k}}{x_{k+1}}{x_i})$ to some element of
$\interpret{\btype}$.

Treating $\eta$ as a substitution, mark statements $\eta
\vdash t \mapsto o$ confirmed if $t\eta = o$.

Now repeat the following steps, until no further changes are made:
\begin{enumerate}
\item\label{alg:call}
  Mark statements $\vdash \apps{\identifier{f}}{e_1}{e_i} \mapsto
  o$ confirmed if $\apps{\identifier{f}}{\ell_1}{\ell_k} = s$ is the
  first clause that matches $\apps{\identifier{f}}{e_1}{e_k}$ and
  $\eta \vdash s \circ x_{k+1} \cdots x_i \mapsto o$ is marked
  confirmed, where $\eta$ is the ``substitution'' such that
  $(\apps{\apps{\identifier{f}}{\ell_1}{\ell_k}}{x_{k+1}}{x_i})\eta
  = \apps{\identifier{f}}{e_1}{e_i}$.
\item\label{alg:var}
  Mark statements $\eta \vdash \apps{x}{s_1}{s_m} \mapsto o$ with
  $x$ a variable confirmed if there is $(e_1,\dots,\linebreak
  e_m,o) \in
  \eta(x)$ s.t. $\eta \vdash s_i \mapsto e_i$ for all $i$.  (By
  immutability, $\apps{x}{s_1}{s_m}$ has base type.)
\item\label{alg:pair}
  Mark statements $\eta \vdash (s_1,s_2) \mapsto (o_1,o_2)$
  confirmed if both $\eta \vdash s_i \mapsto o_i$ are confirmed.
\item\label{alg:ifte}
  Mark statements $\eta \vdash \ifte{s_1}{s_2}{s_3} \mapsto o$
  confirmed if (a) $\eta \vdash s_1 \mapsto \strue$ and $\eta \vdash
  s_2 \mapsto o$ are both confirmed, or (b) $\eta \vdash s_1 \mapsto
  \sfalse$ and $\eta \vdash s_3 \mapsto o$ are both confirmed.
\item\label{alg:choice}
  Mark statements $\eta \vdash \apps{\choice}{s_1}{s_m} \mapsto o$
  confirmed if some $\eta \vdash s_i \mapsto o$ is confirmed.
\item\label{alg:func}
  Mark statements $\eta \vdash \apps{\identifier{f}}{s_1}{s_n}
  \mapsto o$ confirmed if there are $e_1,\dots,e_n$ with
  $\vdash s_i \mapsto e_i$ confirmed for $1 \leq i \leq n$ and
  either (a) $n \geq \arity(\identifier{f})$ and $\apps{\identifier{
  f}}{e_1}{e_n} \mapsto o$ is confirmed, or (b) $n < \arity(
  \identifier{f})$ and $\apps{\identifier{f}}{e_1}{e_m} \mapsto u$ is
  confirmed for all $(e_{n+1},\dots,e_m,u) \in o$.
\end{enumerate}

\textbf{Output:} return the set of all $b$ such that
$\apps{\symb{f}_1}{d_1}{d_M} \mapsto b$ is confirmed.
\end{algorithm}

This algorithm has exponential complexity since, for $\atype$
of order $1$, the cardinality of $\interpret{\atype}$ is exponential
in the input size (the number of constructors in $d_1,\dots,d_M$).
However, since a program with immutable functions cannot
effectively use values with type order $>1$---so can be transformed
to give all values and clauses type order $1$ or $0$---and the size
of each $e \in \interpret{\atype}$ is polynomial, the following
non-deterministic algorithm runs in polynomial time:

\begin{algorithm}\label{alg:np}
Let $S := \{ \asortorpair \mid \asortorpair$ is a type of order $0$
which is used as argument type of some $\identifier{f} \}$.
Let $T := \max\{\Card(\interpret{\asortorpair}) \mid \asortorpair \in
S \}$, and let $N := \langle$number of function symbols$\rangle \cdot
T^{2 \cdot \langle\text{greatest arity}\rangle \cdot \langle
\text{greatest clause depth}\rangle+1}$.
For every clause $\apps{\identifier{f}}{\ell_1}{\ell_k} = s$ of base
type, and every sub-expression of $s$ which has a higher type $\atype$
and is not a variable, generate $N$ elements of $\interpret{\atype}$.
Let $\Xi := \bigcup_{\asortorpair \in S} \interpret{\asortorpair} \cup
\{$ the functional ``values'' thus generated $\}$.

Now run Algorithm~\ref{alg:base}, but only consider statements with
all $e_i$ and $o$ in $\Xi$.
\end{algorithm}

\begin{proposition}\label{prop:simulprogram}
$\prog$ has result value $b$ if{f} there is an evaluation of
Algorithm~\ref{alg:np} which returns a set containing $b$.
\end{proposition}

\begin{proof}[Proof Idea]
We can safely assume that if $\apps{\identifier{f}}{b_1}{b_m} \arrr
d$, it is derived in the same way each time it is used.  Therefore,
in any derivation, at most $T^{\langle\text{greatest arity}+1\rangle}$
distinct values are created to be passed around; the formation of each
value may require $\langle$number of function symbols$\rangle \cdot
T^{\langle\text{greatest arity}\cdot \langle\text{greatest clause
depth}\rangle-1}$ additional helper values.
\end{proof}

By Lemma~\ref{lem:simulmachine} and
Proposition~\ref{prop:simulprogram}, we have: terminating cons-free
programs with immutable functions characterise $\npclass$.
This also holds for the limitation to any data order $\geq 1$.

\section{Beyond $\npclass$}

Unlike Jones, we do not obtain a hierarchy of characterisations for
increasing data orders.  However, we \emph{can} obtain a hierarchical
result by extending the definition of \emph{immutable}:

\begin{definition}
A program has \emph{order $n$ immutable functions} if for all clauses
$\apps{\identifier{f}}{\ell_1}{\ell_k} = s$: (a) the clause uses at
most one variable with a type of order $\geq n$, and (b) if there is
such a variable, then $s$ contains no other sub-expressions with type
order $\geq n$.
\end{definition}

\begin{proposition}\label{prop:final}
Terminating cons-free programs with order $K+1$ immutable functions
characterise $\netime{K}$.
\end{proposition}

\begin{proof}[Proof Idea]
An easy adaptation from the proofs of
Lema~\ref{lem:simulmachine} and Proposition~\ref{prop:simulprogram}.
\end{proof}

\section{Conclusion and discussion}

If Propositions~\ref{prop:simulprogram} and~\ref{prop:final} hold, we
have obtained a characterisation of the hierarchy
$
\npclass \subsetneq \neclass \subsetneq \netime{2} \subsetneq \cdots
\subsetneq \netime{K} \subsetneq \cdots
$ in primarily syntactic terms.
%This is an important step beyond the result in~\cite{kop:sim:17},
%where we did not use immutability and consequently obtained
%$\pclass \subsetneq \elementary \subsetneq \elementary \subsetneq
%\cdots$ or---constraining clauses to disallow partial applications of
%functional variables---the original hierarchy
%$\pclass \subsetneq \etime{1} \subsetneq \etime{2} \subsetneq
%\cdots$.

Arguably, this is a rather inelegant characterisation, both because of
the termination requirement and because the definition of immutability
itself is somewhat arcane; it is not a direct translation of the
intuition that functional values, once created, may not be altered.

The difficulty is that non-determinism is very powerful, and easily
raises expressivity too far when not contained.  This is evidenced
in~\cite{kop:sim:17} where adding non-determinism to cons-free
programs of data order $K\geq 1$ raises the characterised class from
$\etime{K}$ to $\elementary$.
(In~\cite{kop:sim:17}, we did not use immutability; alternatively
restricting the clauses to disallow partial application of functional
variable resulted in the original hierarchy
$\pclass \subsetneq \etime{1} \subsetneq \etime{2} \subsetneq \cdots$.)
In our setting, we must be careful that the allowances made to
\emph{build} the initial function cannot be exploited to manipulate
exponentially many distinct values.
For example, if we drop the termination requirement, we could identify
the lowest number $i < 2^n$ such that $P(i)$ holds for any
polytime-decidable property $P$, as follows:

\begin{tabular}{ll}
$\symb{bit\_of\_lowest}\ \numrep{n}\ \numrep{j} =
  \symb{f}\ \numrep{n}\ \numrep{j}$ &
$\symb{nul}\ \numrep{j} = \sfalse$ \\
$\symb{f}\ \numrep{n} = \choice\ \symb{nul}\ 
  (\symb{succtest}\ (\symb{f}\ \numrep{n}))$ &
$\symb{succ}\ F\ \numrep{n}\ \numrep{j} = \dots$ \\
\extend{$\symb{succtest}\ F\ \numrep{j} = \ifte{\:\symb{prop}\ F\:}{\:
  F\ \numrep{j}\:}{\:\symb{succ}\ F\ \numrep{n}\ \numrep{j}}$} \\
\end{tabular}

\noindent
Here, the clauses for $\symb{succ}\ F\ \numrep{n}\ \numrep{j}$ result
in $\strue$ if $b_j = 1$ when representing the successor of $F$ as a
bitvector $b_1\dots b_n$, and in $\sfalse$ otherwise.  It is unlikely
that the corresponding decision problem is in $\npclass$.  Similar
problems may arise if we allow multiple higher-order variables,
although we do not yet have an example illustrating this problem.

In the future, we intend to complete the proofs, and study these
restrictions further.  Even if this does not lead to an elegant
characterisation of the $\netime{K}$ hierarchy, it is likely to give
further insights in the power of non-determinism in higher-order
cons-free programs.
%, and we might for instance obtain new
%characterisations of the space complexity classes.

\bibliography{references}
\bibliographystyle{plainurl}

\end{document}